\documentclass{lmcs}
\keywords{MANDATORY list of keywords}

\usepackage{hyperref}
\usepackage{stackrel}
\usepackage{lmcs-preamble,tikz-cd}
\usetikzlibrary{matrix}

\renewcommand{\epsilon}{\varepsilon}
\renewcommand{\phi}{\varphi}

\newcommand{\dfib}{\mathit{DFib}}

\theoremstyle{plain} 

\begin{document}

\title{Simple game semantics and Day convolution}

\author[C.~Eberhart]{Clovis Eberhart}	
\address{National Institute of Informatics, Japan}	

\author[T.~Hirschowitz]{Tom Hirschowitz}	
\address{Univ. Grenoble Alpes, Univ. Savoie Mont Blanc, CNRS, LAMA}	

\author[A.~Laouar]{Alexis Laouar}	
\address{Univ. Grenoble Alpes, Univ. Savoie Mont Blanc, CNRS, LAMA}	





\begin{abstract}
  \noindent Game semantics has provided adequate models for a variety
  of programming languages~\cite{hyland97}, in which types are
  interpreted as two-player games and programs as strategies.
  Melliès~\cite{PAM} suggested that such categories of games and
  strategies may be obtained as instances of a simple abstract
  construction on weak double categories.  However, in the particular
  case of simple games~\cite{DBLPconf/lics/HarmerHM07}, his
  construction slightly differs from the standard category.  We refine
  the abstract construction using factorisation systems, and show that
  the new construction yields the standard category of simple games
  and strategies. Another perhaps surprising instance is Day's
  \emph{convolution} monoidal structure on the category of presheaves
  over a strict monoidal category.
\end{abstract}

\maketitle

\section{Introduction}

The construction of most game models follows a common pattern.
Typically, a function $A \to  B$ is interpreted as a strategy on a
compound game made of $A$ and $B$, where the program plays as
\emph{Proponent} (P) on $B$ and as \emph{Opponent} (O) on $A$.  A
crucial feature of game models is composition of strategies, by which
two strategies $\sigma \colon  A\to B$ and $\tau \colon B\to C$ yield a strategy $\tau \circ \sigma \colon A \to 
C$. Intuitively, $\tau \circ \sigma $ lets $\sigma $ and $\tau $ interact on $B$ until one of
them produces a move in $A$ or $C$. However, in order to obtain a
strategy $A \to  C$, everything occurring on $B$ should be hidden.

Although widely acknowledged, this strong commonality is also
recognised as poorly understood, particularly in the presence of
\emph{innocence}, a constraint on strategies that restricts them to
purely functional behaviour.  This has prompted a number of attempts
at clarifying the
situation~\cite{hirschowitz2007theory,DBLPconf/lics/HarmerHM07,DBLPconf/lics/EberhartH18}.
In particular, Melliès~\cite{PAM} recently proposed a novel
explanation, of unprecedented simplicity. Indeed, it rests upon a
purely categorical construction, essentially taking the slice of a
weak double category over an internal monad. This suggests that the
approach may encompass a wide variety of game models, which is
currently not the case as it is restricted to linear languages.  More
generally, it may prompt new connections between game semantics and
other settings.

In this paper, we focus on a peculiar feature of Melliès's approach,
namely the slight discrepancy between his category of games and
strategies and the standard one. Indeed, while standard strategies
only play one move at a time, Melliès's may play several moves
simultaneously. This raises the question of whether the standard
setting may be recovered by refining his approach.  We answer this
question positively by enriching the setting with a factorisation
system~\cite{Bousfield}: using a double categorical variant of the
so-called comprehensive factorisation system~\cite{comprehensive}, we
obtain the standard setting as an instance.

As a bonus, one of the connections alluded to above is established.
Namely, we show that the Day \emph{convolution} product~\cite{Day}
arises as an instance of our refined framework, though only in the
restricted case of \emph{strict} monoidal categories.  The convolution
product, which arose in algebraic topology, extends the monoidal
structure of a given category $\CCC $ to the category $\psh{\CCC }$ of
\emph{presheaves} on $\CCC $, i.e., contravariant functors $\op{\CCC } \to  \Set $.
This makes formal the similarity, noted in~\cite[Section~6.5]{Clovis},
between convolution and composition of strategies, by showing that
both are instances of the same construction.

\subsection*{Related work}
Beyond Melliès's obviously related work
and~\cite{hirschowitz2007theory}, \citet{GarnerShulman} prove results
related to Theorems~\ref{thm:basic} and~\ref{thm:facto}.  The common
ground for comparison is the restriction of our
Theorem~\ref{thm:basic} to weak double categories with trivial
vertical category, i.e., monoidal categories. Their Theorem~14.2 is a
generalisation in another direction, namely that of monoidal
bicategories, and their Theorem~14.5 could in particular accomodate
various sorts of bicategorical factorisation systems.  What is needed
for dealing with Day convolution in full generality (as mentioned in
Section \ref{sec:conc}) is a common generalisation of their
Theorem~14.5 and our Theorem~\ref{thm:facto}.

\subsection*{Plan}
In Section~\ref{sec:mellies}, after briefly reviewing \emph{double
  categories}~\cite{Ehresmanndouble}, a 2-dimensional generalisation
of categories, we recall the cornerstone of Melliès's variant of
simple games, the double category $\moo$. We then explain how
Melliès's bicategory of simple games may be obtained by a double
categorical generalisation of the slice construction.  In
Section~\ref{sec:simple}, we then refine this construction. We
introduce \emph{double factorisation systems} and show that
restricting a slice weak double category to members of the right class
of a double factorisation system again yields a weak double category.
Finally, we observe that this construction has both standard simple
games and Day convolution as instances.  We conclude and give some
persective on future work in Section~\ref{sec:conc}.

\section{Melliès's simple games}\label{sec:mellies}
The cornerstone of Melliès's account of simple games is a double
category $\moo$ (called the \emph{clock}) which embodies the essence
of scheduling.  So let us briefly recall what a double category is,
and then describe $\moo$.

\subsection{Recap on double categories}
A double category $\CCC$ essentially consists of two categories $\CCC_h $ and $\CCC_v $
sharing the same object set, together with a set of \emph{cells}
  \begin{center}
    \diag{%
      A \& B \\
      C \& D\rlap{,} %
    }{%
      (m-1-1) edge[pro,twoa={f}] (m-1-2) %
      edge[labell={u}] (m-2-1) %
      (m-2-1) edge[pro,twob={g}] (m-2-2) %
      (m-1-2) edge[labelr={v}] (m-2-2) %
      (a) edge[cell=0,labelr={\alpha }] (b) %
    }
  \end{center}
  where $A,B,C$, and $D$ are objects, $f$ and $g$ are morphisms in
  $\CCC_h $, and $u$ and $v$ are morphisms in $\CCC_v $.  In order to
  distinghish notationally between horizontal and vertical morphisms,
  we mark horizontal ones with a bullet. Cells are furthermore
  equipped with composition and identities in both directions.  E.g.,
  to any given cells with compatible vertical border as below left is
  assigned a composite cell as below right
  \begin{equation}
    \diag{%
      A \& B \& E \\
      C \& D \& F %
    }{%
      (m-1-1) edge[pro,twoa={S}] (m-1-2) %
      edge[labell={p}] (m-2-1) %
      (m-2-1) edge[pro,twob={T}] (m-2-2) %
      (m-1-2) edge[labelr={q}] (m-2-2) %
      (a) edge[cell=0,labelr={\alpha }] (b) %
      (m-1-2) edge[pro,twoabove={a'}{S'}] (m-1-3) %
      (m-2-2) edge[pro,twobelow={b'}{T'}] (m-2-3) %
      (m-1-3) edge[labelr={r}] (m-2-3) %
      (a') edge[cell=0,labelr={\beta }] (b') %
    }
    \hfil \mapsto  \hfil
    \diag{%
      A \& E \\
      C \& F, %
    }{%
      (m-1-1) edge[pro,twoa={S' \bullet S}] (m-1-2) %
      edge[labell={p}] (m-2-1) %
      (m-2-1) edge[pro,twob={T' \bullet T}] (m-2-2) %
      (m-1-2) edge[labelr={r}] (m-2-2) %
      (a) edge[cell=0,labelr={\beta  \bullet \alpha }] (b) %
    }
    \label{eq:compoH}
  \end{equation}
 Similarly, we have horizontal identities
  \begin{center}
    \diag{%
      A \& A \\
      B \& B, %
    }{%
      (m-1-1) edge[pro,twoa={id^\bullet_A}] (m-1-2) %
      edge[labell={p}] (m-2-1) %
      (m-2-1) edge[pro,twob={id^\bullet_B}] (m-2-2) %
      (m-1-2) edge[labelr={p}] (m-2-2) %
      (a) edge[cell=0,labelr={id^\bullet_p}] (b) %
    }
  \end{center}
  Both notions of composition are required to be associative and the
  corresponding identities unital.  Finally, the \emph{interchange
    law} requires the two different ways of parsing any compatible pasting
  \begin{center}
    \diag{
      A \& B \& C \\
      D \& E \& F \\
      G \& H \& I 
    }{%
      (m-1-1) edge (m-2-1)
      (m-2-1) edge (m-3-1)
      (m-1-2) edge (m-2-2)
      (m-2-2) edge (m-3-2)
      (m-1-3) edge (m-2-3)
      (m-2-3) edge (m-3-3)
      (m-1-1) edge[pro] node[coordinate] (ul) {} (m-1-2)
      (m-1-2) edge[pro] node[coordinate] (ur) {} (m-1-3)
      (m-2-1) edge[pro] node[coordinate] (ml) {} (m-2-2)
      (m-2-2) edge[pro] node[coordinate] (mr) {} (m-2-3)
      (m-3-1) edge[pro] node[coordinate] (dl) {} (m-3-2)
      (m-3-2) edge[pro] node[coordinate] (dr) {} (m-3-3)
      (ul) edge[cell=0.2, labelr={\alpha }] (ml)
      (ml) edge[cell=0.2, labelr={\beta }] (dl)
      (ur) edge[cell=0.2, labelr={\gamma }] (mr)
      (mr) edge[cell=0.2, labelr={\delta }] (dr)
    }%
  \end{center}
  to agree, i.e.,
    \begin{center}
      $(\delta  \circ  \gamma ) \bullet (\beta  \circ  \alpha ) = (\delta  \bullet \beta ) \circ  (\gamma  \bullet \alpha )$.
    \end{center}

    There is an alternative point of view on double categories which
    will be crucial to us. The previous presentation has emphasised
    $\CCC_v $ and $\CCC_h $, and added the set of cells. But we could also put
    forward $\CCC_V$, the category whose objects are horizontal
    morphisms, and whose morphisms are cells.
    Indeed, double categories may be axiomatised based on a pair of functors
    \begin{equation}
    l,r \colon  \CCC_V \rightrightarrows \CCC_v .\label{eq:lr}
  \end{equation}

    What is missing from this is just horizontal composition and
    identities (for horizontal arrows and cells), which may be
    postulated by requiring that $l$ and $r$ form a \emph{category
      object} in $\Cat$.

    Equivalently, still, we can consider the following structure,
    which is almost a (large) double category.
    \begin{defi}\label{def:span:cat}
      Let $\Span(\Cat)$ have
      \begin{itemize}
      \item as objects all small categories,
      \item as vertical morphisms all functors,
      \item as horizontal morphisms $A \proto B$ all spans
        $A \ot  C \to  B$ of functors, and
      \item as cells 
        \begin{center}
          \diag{A \& B \\ U \& V}{
            (m-1-1) edge[pro,twoa={}] (m-1-2) %
            edge[labell={}] (m-2-1) %
            (m-2-1) edge[pro,twob={}] (m-2-2) %
            (m-1-2) edge[labelr={}] (m-2-2) %
            (a) edge[cell=0] (b) %
            }
          all commuting diagrams
          \diag{%
            A \& C \& B \\
            U \& W \& V
          }{%
            (m-1-1) edge[<-,labela={}] (m-1-2) %
            edge[labell={}] (m-2-1) %
            (m-2-1) edge[<-,labelb={}] (m-2-2) %
            (m-1-2) edge[labelr={}] (m-2-2) %
            (m-1-2) edge[labela={}] (m-1-3) %
            (m-2-2) edge[labelb={}] (m-2-3) %
            (m-1-3) edge[labelr={}] (m-2-3) %
          }
        \end{center}
        in $\Cat$.
      \end{itemize}
      Vertical composition is given by (componentwise) composition of
      functors, while horizontal composition is given by pullbacks and
      their universal property.
  \end{defi}
  The structure formed by $\Span(\Cat)$ is a \emph{weak double
    category}~\cite{GarnerPhD}, a weak form of double category where
  horizontal composition is only associative and unital up to coherent
  isomorphism, in a suitable sense. In particular, the horizontal
  arrows and \emph{special} cells of a weak double category form a
  bicategory, where special means that the left and right borders are
  identities. This entails that we may define monads in weak double
  categories just like one usually does in bicategories, and we have:
  \begin{prop}
    A double category is the same as a monad in $\Span(\Cat)$.
    \label{prop:dcat_monad}
  \end{prop}
  Indeed, composing a span~\eqref{eq:lr} with itself in $\Span(\Cat)$
  amounts to constructing the category of pairs of compatible
  horizontal morphisms and cells~\eqref{eq:compoH}, so requiring a
  monad multiplication is requiring horizontal composition.

\subsection{The clock}
The vertical category of Melliès's $\moo$ is the free category on the graph
\begin{center}
  \diag{
    O \& P.
  }{%
    (m-1-1) edge[bend left] (m-1-2) %
    (m-1-2) edge[bend left] (m-1-1) %
  }
\end{center}
So its objects are just $O$ and $P$, which stand for \emph{Opponent}
and \emph{Proponent}, as in game models, and morphisms just count the
number of (alternating) moves between them.  One way to denote such
morphisms is just as alternating strings of $O$s and $P$s.

The horizontal category $\moo _h $ is simply the ordinal $\two$, viewed as
a category, except that $0$ is here renamed to $O$ and $1$ to $P$, so
we have $O \leq  P$, and there are only three horizontal morphisms: $OO$,
$PP$, and $OP$.

Finally, cells describe the allowed schedulings in an arrow game: they
are simply arrows in the free category $\moo_V$ over the famous state
diagram for simple games
\begin{center}
  \diag{%
    OO \& OP \& PP.
  }{%
    (m-1-1) edge[bend left] (m-1-2)
    edge[bend right,<-] (m-1-2)
    (m-1-2) edge[bend left] (m-1-3)
    edge[bend right,<-] (m-1-3)
  }
\end{center}
Again, a way to write such morphisms is by valid sequences in
$\ens{OO,OP,PP}$.  Vertical composition of cells is simply composition
in $\moo_V$, i.e., concatenation. The most intuitive way to introduce
horizontal composition is to depict basic cells as triangles
\begin{center}
  \diag{%
    O \& O \\
    \& P
  }{%
    (m-1-1) edge[pro] (m-1-2) %
    edge[pro] (m-2-2) %
    (m-1-2) edge (m-2-2) %
  }
  \hfil 
  \diag{%
    O \& P \\
    P \& 
  }{%
    (m-1-1) edge[pro] (m-1-2) %
    edge (m-2-1) %
    (m-2-1) edge[pro] (m-1-2) %
  }
  \hfil 
  \diag{%
    P \& P \\
    O \& 
  }{%
    (m-1-1) edge[pro] (m-1-2) %
    edge (m-2-1) %
    (m-2-1) edge[pro] (m-1-2) %
  }
  \hfil
  \diag{%
    O \& P \\
    \& O\rlap{.}
  }{%
    (m-1-1) edge[pro] (m-1-2) %
    edge[pro] (m-2-2) %
    (m-1-2) edge (m-2-2) %
  }
\end{center}
General cells are obtained by stacking up such basic triangles.
\begin{example}
  To see the connection with standard game models, consider a typical
  play like below left, whose scheduling is the cell below right in
  $\moo$. Written as a sequence, this morphism is $(OO,OP,PP,OP,OO)$ —
  the sequence of involved horizontal morphisms.
  \begin{center}
    \Diag(0.2,1){
      \path[->] (m-1-1) edge (m-1-2);
    }{
      \Nat  \& \Nat  \\
      \& q \\
      q \\ \\
      6 \\
      \& 12
    }{
      (m-3-1) edge (m-2-2)
      (m-5-1) edge[bend left] (m-3-1)
      (m-6-2) edge[bend right] (m-2-2)
    }
    \hfil
    \diag{%
      O \& O \\
      P \& P \\
      O \& O
    }{%
      (m-1-1) edge[pro] (m-1-2) %
      edge[pro] (m-2-2) %
      edge (m-2-1) %
      (m-1-2) edge (m-2-2) %
      (m-2-1) edge[pro] (m-2-2) %
      (m-2-1) edge[pro] (m-2-2) %
      edge (m-3-1) %
      (m-3-1) edge[pro] (m-2-2) %
      (m-2-2) edge (m-3-2) %
      (m-3-1) edge[pro] (m-3-2) %
    }
  \end{center}
  The first $q$ move corresponds to the top triangle. The non-trivial
  vertical side of the latter shows that the move in question is
  played on the right-hand game.  The next $q$ move corresponds to the
  next triangle downwards, and so on.
\end{example}
Depicting cells as stacks of triangles yields the following inductive
definition of horizontal composition:
\begin{itemize}
\item If there is an ‘outwards’ bottom triangle, i.e., the bottom of $\alpha $ and $\beta $
  look like either of
  \begin{center}
    \diag{%
      O \& O \& X \\
        \&   \& X^\bot  
      }{%
        (m-1-1) edge[pro] (m-1-2) %
        (m-1-2) edge[pro] (m-1-3) %
        edge[pro] (m-2-3) %
        (m-1-3) edge (m-2-3) %
    }
    \hfil
        \diag{%
      X \& P \& P \\
       X^\bot  
      }{%
        (m-1-1) edge[pro] (m-1-2) %
        (m-1-2) edge[pro] (m-1-3) %
        (m-1-1) edge (m-2-1) %
        (m-2-1) edge[pro] (m-1-2) %
      }
  \end{center}
  with $X \in  \ens{O,P}$ and $X^\bot  $ denoting the other player, then the
  composite is obtained by composing the rest of $\alpha $ and $\beta $, and
  appending the obvious triangle ($(OX,OX^\bot  )$, resp.\ $(XP,X^\bot  P)$).
\item Otherwise, there is a pair of interacting bottom triangles, as
  in
  \begin{center}
    \diag{%
      O \& X \& P \\
        \& X^\bot  \rlap{,}
      }{%
        (m-1-1) edge[pro] (m-1-2) %
        edge[pro] (m-2-2) %
        (m-1-2) edge[pro] (m-1-3) %
        (m-2-2) edge[pro] (m-1-3) %
        (m-1-2) edge (m-2-2) %
    }
  \end{center}
\end{itemize}
in which case the composite is simply the composite of the rest of $\alpha $
and $\beta $ -- which is precisely where game semantical \emph{hiding} is
encoded in $\moo$.
\begin{rem}
  Ambiguous configurations as below, where we would not know which
  triangle to put last in the composite, cannot occur. Indeed,
  existence of the left-hand triangle forces $M = P$, while existence
  of the right-hand one forces $M = O$.
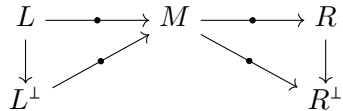
  \begin{center}
    \diag{%
      L  \& M \& R \\
      L^\bot   \&   \& R^\bot  
      }{%
        (m-1-1) edge[pro] (m-1-2) %
        edge (m-2-1) %
        (m-1-2) edge[pro] (m-1-3) %
        edge[pro] (m-2-3) %
        (m-1-3) edge (m-2-3) %
        (m-2-1) edge[pro] (m-1-2) %
    }

  \end{center}  
\end{rem}

Melliès obtains:
\begin{prop}
  Composition as just defined makes $\moo$ into a double category.
  \label{prop:moo_dcat}
\end{prop}

\subsection{Simple games and strategies}
This is where Melliès's approach is novel. He simply puts:
\begin{defi}
  A \emph{game} is a functor to $\moo _v $, and a \emph{strategy}
  from $p\colon  A \to  \moo _v $ to  $q\colon  B \to  \moo _v $ is a commuting diagram
  \begin{equation}
    \diag{%
      A \& S \& B \\
      \moo _v  \& \moo_V \& \moo _v 
    }{%
      (m-1-1) edge[<-,labela={}] (m-1-2) %
      edge[labell={p}] (m-2-1) %
      (m-2-1) edge[<-,labelb={}] (m-2-2) %
      (m-1-2) edge[labelr={}] (m-2-2) %
      (m-1-2) edge[labela={}] (m-1-3) %
      (m-2-2) edge[labelb={}] (m-2-3) %
      (m-1-3) edge[labelr={q}] (m-2-3) %
    }
    \label{eq:strat}
  \end{equation}
  in $\Cat$.
\end{defi}
Intuitively, in a game $p\colon  A \to  \moo_ v$, the objects of $A$ are plays,
and $p$ indicates which player should play next. Morphisms are
sequences of moves, whose number and polarity is again indicated by
$p$. The notion of strategy may be understood as follows: the limit
  \begin{center}
    \Diag{%
      \pbk[2em]{m-2-1}{m-1-2}{m-2-3} %
    }{%
      \& \P _{A,B} \&  \\
      A \&  \& B \\      
      \moo _v  \& \moo_V \& \moo _v 
    }{%
      (m-2-1) edge[<-,labela={}] (m-1-2) %
      edge[labell={p}] (m-3-1) %
      (m-3-1) edge[<-,labelb={}] (m-3-2) %
      (m-1-2) edge[labelr={}] (m-3-2) %
      (m-1-2) edge[labela={}] (m-2-3) %
      (m-3-2) edge[labelb={}] (m-3-3) %
      (m-2-3) edge[labelr={q}] (m-3-3) %
    }
  \end{center}
  may be thought of as a category of plays on the arrow game $A \to  B$,
  and the induced map $S \to  \P _{A,B}$ describes which plays are accepted
  by the strategy.

  Recalling the weak double category $\Span(\Cat)$ from
  Definition~\ref{def:span:cat}, we observe that a game is a morphism
  to $\moo _v $ in $\Span(\Cat)_v $, while a strategy $A \to  B$ is merely a
  cell
  \begin{center}
    \diag{%
      A \& B \\
      \moo _v  \& \moo _v .
    }{%
      (m-1-1) edge[pro,twoa={S}] (m-1-2) %
      edge[labell={p}] (m-2-1) %
      (m-2-1) edge[pro,twob={\moo_V}] (m-2-2) %
      (m-1-2) edge[labelr={q}] (m-2-2) %
      (a) edge[cell=0] (b) %
    }
  \end{center}
  This provides a simple way of composing strategies, using the monad
  structure of $\moo$ (given by Propositions~\ref{prop:dcat_monad}
  and~\ref{prop:moo_dcat}): the composite of $S\colon  A \to  B$ and $T\colon  B \to  C$
  is simply the pasting
  \begin{equation}
    \diag{%
      A \& B \& C \\
      \moo _v  \& \moo _v  \& \moo _v  %
    }{%
      (m-1-1) edge[pro,twoa={S}] (m-1-2) %
      edge[labell={p}] (m-2-1) %
      (m-2-1) edge[pro,twob={\moo_V}] (m-2-2) %
      edge[bend right,pro,twobelow={bbb}{\moo_V}] (m-2-3) %
      (m-1-2) edge[labelr={q}] (m-2-2) %
      (m-1-2) edge[pro,twoabove={aa}{T}] (m-1-3) %
      (m-2-2) edge[pro,twobelow={bb}{\moo_V}] (m-2-3) %
      (m-1-3) edge[labelr={r}] (m-2-3) %
      (a) edge[cell=0] (b) %
      (aa) edge[cell=0] (bb) %
      (m-2-2) edge[celllr={0}{-0.1},labelr={\ \mu }] (bbb) %
    }
    \label{eq:compo:slice}
  \end{equation}
  where $\mu $ denotes the monad multiplication for $\moo$, i.e.,
  horizontal composition of cells.

  Strategies are thus equipped with weak double category structure
  (although Melliès only considers the underlying bicategory) by
  applying the following general result to the monad $\moo$ in
  $\Span(\Cat)$:
  \begin{thm}\label{thm:basic}
    Given any monad $M_V\colon  M_v  \proto M_v $ in a weak double category $\CCC$,
    there is a \emph{slice} weak double category $\CCC/M$ whose
    \begin{itemize}
    \item vertical category is $(\CCC/M)_v  = \CCC_v /M_v $;
    \item vertical category of cells is $(\CCC/M)_V = \CCC_V/M_V$;
    \item horizontal composition of cells is given by pasting
      with monad multiplication, as in~\eqref{eq:compo:slice}; and
    \item horizontal identity on $p\colon  A \to  M_v $ is the pasting
      \begin{equation}
        \diag(.6,2){%
          A \& A \\
          M_v  \& M_v . %
        }{%
          (m-1-1) edge[pro,twoa={}] (m-1-2) %
          edge[labell={p}] (m-2-1) %
          (m-2-1) edge[pro,identity,twob={}] node[coordinate] (aa) {} (m-2-2) %
          edge[pro,bend right=40,labelb={M_V}] node[coordinate] (bb) {} (m-2-2) %
          (m-1-2) edge[labelr={p}] (m-2-2) %
          (a) edge[cell=0,labelr={id^\bullet_p }] (b) %
          (aa) edge[cell=0.1,labelr={\eta }] (bb) %
        }
        \label{eq:identity:slice}
      \end{equation}
    \end{itemize}
  \end{thm}

  \begin{prop}
    Melliès's bicategory of simple games is the underlying bicategory
    of $\Span(\Cat)/\moo$.
  \end{prop}

  A weak double category with trivial vertical category is nothing but
  a monoidal category.  In that case, the theorem reduces to the
  following well-known result used, e.g., by~\citet{Webergenmorph}:
  \begin{cor}
    The slice of a monoidal category over a monoid is again monoidal.
  \end{cor}

  \section{Recovering simple games}\label{sec:simple}
  \subsection{Restricting to discrete fibrations}
  The present work was prompted by Melliès's observation that
  $\Span(\Cat) / \moo$ is not quite equivalent to the standard category
  of simple games and strategies.  Indeed, one might have expected
  that, restricting to strategies whose underlying functor
  $S \to  \P _{A,B}$ is an inclusion, we would obtain an equivalent
  category. But this is not the case: Melliès's games and strategies
  are intrinsically more general. This is even emphasised as a
  feature, as it has the advantage of smoothing things up in the
  context of asynchronous games, where a different clock is used.  In
  our sequential setting, the extra generality is twofold. First, the
  considered games may have no time origin — concretely there is no
  empty play.
  \begin{example}
    A symptom is that the categories $A$ may not be well founded,
    i.e., they may contain infinite chains $… \to  a_n \to  … \to  a_0 $.
  \end{example}
  The second source of extra generality is that games may feature
  ‘compound moves’, i.e., indecomposable morphisms whose image in
  $\moo _v $ has length $>1$.
  \begin{example}
    Take for $A$,
    e.g., the ordinal $\one$ viewed as a category, and map its unique
    morphism to $OPOP$ in $\moo _v $.
  \end{example}
  This thus raises the question: can standard simple games be
  recovered by refining the abstract Theorem~\ref{thm:basic}?

  The first step towards this is to characterise the games
  $A \to  \moo _v $ that correspond to standard simple games.  This is
  easy: by definition, standard simple games are trees, which may be
  defined as presheaves over the ordinal $\omega $. But presheaves are
  equivalent to discrete fibrations, hence the idea of restricting
  $(\Span(\Cat)/\moo)_v $ to discrete fibrations. However, this does not
  quite work, as $\moo _v $ lacks a ‘time origin’: presheaves on
  $\moo _v $ describe games in which there is no first move — all moves
  have predecessors. But if we slice $\moo _v $ under $O$, then we get
  it right, as we have $O/\moo_v  \iso  \omega $. Similarly, $OO/\moo_V$ describes
  scheduling in the arrow category starting from $OO$.

  Funnily enough, $OO\colon  O \proto O$, viewed as a horizontal
  endomorphism in $\moo$, is a comonad (the identity comonad on $O$),
  hence we may apply Theorem~\ref{thm:basic} to obtain:
  \begin{cor}
    The slice $\H := OO/\moo$ forms a double category.
  \end{cor}
  \begin{rem}
    We need in fact a variant of Theorem~\ref{thm:basic} for strict
    double categories.
  \end{rem}
  Taking $\H$ (standing for \emph{horloge}, french for clock) as a
  replacement for $\moo$, we get the desired property that discrete
  fibrations over $\H_v $ are equivalent to $\psh{\omega }$.

  In a similar vein, in recent work on game
  semantics~\cite{HirschoDoubleCats,OngTsukada,DBLPconf/lics/EberhartH18},
  a concurrent notion of strategy was defined as presheaves on plays.
  It thus would seem natural to also restrict
  strategies~\eqref{eq:strat} to ensure that the induced functor
  $S \to  \P _{A,B}$ is a discrete fibration. This may be enforced
  directly:
  \begin{lem}
    Given a commuting diagram of functors
    \begin{center}
      \diag{%
        A \& S \& B \\
        X \& T \& Y
      }{%
        (m-1-1) edge[<-,labela={}] (m-1-2) %
        edge[labell={p}] (m-2-1) %
        (m-2-1) edge[<-,labelb={l}] (m-2-2) %
        (m-1-2) edge[labelr={m}] (m-2-2) %
        (m-1-2) edge[labela={}] (m-1-3) %
        (m-2-2) edge[labelb={r}] (m-2-3) %
        (m-1-3) edge[labelr={q}] (m-2-3) %
      }
    \end{center}
    where $p$ and $q$ are discrete fibrations, letting $P$ denote the
    limit of the subdiagram $$A \to  X \ot  T \to  Y \ot  B,$$ the induced functor
    $S \to  P$ is a discrete fibration iff the middle functor $m\colon  S \to  T$
    is.
  \end{lem}
  \begin{proof}
    Discrete fibrations are the right class of a (strong)
    factorisation system (see Lemma~\ref{lem:spans:facto} below),
    hence are stable under composition and pullback, and furthermore
    enjoy left cancellation: if $g \circ  f$ and $g$ are discrete
    fibrations, then so is $f$.

    Now, the limit $P$ may be computed by taking pullbacks of $p$ and
    $q$, respectively along $l$ and $r$, and then taking the pullback
    of the obtained cospan.  Thus, if $p$ and $q$ are discrete
    fibrations, then by stability under pullback and composition, so
    is the projection functor $P \to  T$.

    Thus, if the induced functor $S \to  P$ is a discrete
    fibration, then so is the middle functor $S \to  T$ by stability
    under composition.

    Conversely, if the middle functor $S \to  T$ is a discrete
    fibration, then so is the induced functor $S \to  P$ by left
    cancellation.
  \end{proof}

  \subsection{Simple games}
  We thus hope to recover standard simple games by slicing $\moo$
  under $OO$, and restricting the slice construction to discrete
  fibrations (for vertical morphisms and cells).  This may be carried
  over to the abstract setting using the observation, recalled in the
  above proof, that discrete fibrations are the right class of a
  factorisation system.
  \begin{defi}\label{def:double:facto}
    A \emph{double factorisation system} on a weak double category $\CCC$
    consists of factorisation systems $(\LLL_v ,\RRR _v )$ and $(\LLL_V,\RRR _V)$ on
    $\CCC_v $ and $\CCC_V$, respectively, such that
    \begin{enumerate}[label=(\emph{\Alph*})]
    \item \label{double:facto:L} $\LLL_V$ is preserved under
      horizontal composition and contains horizontal identities, and
    \item
      \label{double:facto:split} all cells $\alpha $ as below left with
      $\ell ,\ell ' \in  \LLL_v $ and $r,r' \in  \RRR _v $ factor as below right, with
      $\lambda  \in  \LLL_V$ and $\rho  \in  \RRR _V$.
      \begin{center}
        \diag{%
          A \& A' \\
          B \& B' \\
          C \& C' 
        }{%
          (m-1-1) edge[pro,twoa={S}] (m-1-2) %
          edge[labell={\ell }] (m-2-1) %
          (m-1-2) edge[labelr={\ell '}] (m-2-2) %
          (m-2-1) edge[labell={r}] (m-3-1) %
          (m-3-1) edge[pro,twob={U}] (m-3-2) %
          (m-2-2) edge[labelr={r'}] (m-3-2) %
          (a) edge[cell=0,labelr={\alpha }] (b)
        }
        \hfil $=$ \hfil
        \diag{%
          A \& A' \\
          B \& B' \\
          C \& C' 
        }{%
          (m-1-1) edge[pro,twoa={S}] (m-1-2) %
          edge[labell={\ell }] (m-2-1) %
          (m-2-1) edge[pro,twob={},labelaat={T}{.8},two={aa}{below}] (m-2-2) %
          (m-1-2) edge[labelr={\ell '}] (m-2-2) %
          (m-2-1) edge[labell={r}] (m-3-1) %
          (m-3-1) edge[pro,twobelow={bb}{U}] (m-3-2) %
          (m-2-2) edge[labelr={r'}] (m-3-2) %
          (a) edge[cell=0,labelr={\lambda }] (b)
          (aa) edge[cell=0,labelr={\rho }] (bb)          
        }

      \end{center}
    \end{enumerate}
  \end{defi}

  \begin{lem}\label{lem:spans:facto}
    Discrete fibrations and componentwise discrete fibrations are the
    right classes of factorisation systems which together form a
    double factorisation system for $\Span(\Cat)$.
  \end{lem}
  \begin{proof}
    It is well-known that discrete fibrations may be defined by unique
    lifting w.r.t.\ the injection $\one \into  \two$ mapping $0$ to $1$.  The
    only non-obvious point is then that componentwise discrete
    fibrations are stable under horizontal composition, which follows
    from the general fact that the right class of any factorisation
    system is stable under pullback in the arrow category.
  \end{proof}

  This leads us to the following generalisation of
  Theorem~\ref{thm:basic}:
  \begin{thm}\label{thm:facto}
    Given any monad $M_V\colon  M_v  \proto M_v $ in a weak double category $\CCC$
    with double factorisation system $((\LLL_v ,\RRR _v ),(\LLL_V,\RRR _V))$, there is a
    \emph{slice} weak double category $\CCC/_\RRR M$ whose
    \begin{itemize}
    \item vertical category $(\CCC/_\RRR M)_v $ is $\CCC_v /_{\RRR _v } M_v $, the full
      subcategory of $\CCC_v /M_v $ on maps in $\RRR _v $;
    \item vertical category of cells is $(\CCC/_\RRR  M)_V = \CCC_V/_{\RRR _V} M_V$;
    \item horizontal composition of cells is given by factoring the
      pasting~\eqref{eq:compo:slice} as $\rho  \circ  \lambda $ and returning $\rho $;
    \item and whose horizontal identity on any $p\colon  A \to  M_v $ is given by
      factoring~\eqref{eq:identity:slice} as $\rho  \circ  \lambda $ and returning
      $\rho $.
    \end{itemize}
  \end{thm}
  \begin{rem}
    The fact that identities and horizontal composites have the right
    perimeter follows from Condition~\ref{double:facto:split} in
    Definition~\ref{def:double:facto}.
  \end{rem}

  \begin{rem}
    When the monad multiplication is in $\RRR _V$, and $\RRR _V$ is stable
    under horizontal composition and contains horizontal identities,
    then so are~\eqref{eq:compo:slice} and~\eqref{eq:identity:slice},
    hence $\CCC/_\RRR M$ is a sub weak double category of $\CCC/M$.
  \end{rem}
  In the general case, we may picture composition in $(\CCC/_\RRR M)_h$ as
  follows:
  \begin{center}
    \diag(.7,1.5){%
      A \& B \& C \\
      M _v  \& M _v  \& M _v \rlap{.}
    }{%
      (m-1-1) edge[pro,twoa={S}] (m-1-2) %
      edge[labell={p}] (m-2-1) %
      (m-2-1) edge[pro,twob={M_V}] (m-2-2) %
      edge[bend right,pro={.6},labelbat={M_V}{0.6}] node[coordinate,pos=.6] (bbb) {} (m-2-3) %
      (m-1-2) edge[labellat={q}{.3}] (m-2-2) %
      (m-1-2) edge[pro,twoabove={aa}{T}] (m-1-3) %
      (m-2-2) edge[shorten <={-1ex},pro,twobelow={bb}{M_V}] (m-2-3) %
      (m-1-3) edge[labelr={r}] (m-2-3) %
      (a) edge[cell=0] (b) %
      (aa) edge[cell=0] (bb) %
      (m-2-2) edge[celllr={0}{0.1},labelbl={\ \mu }] (bbb) %
      (m-1-1) edge[pro={.6},fore,bend right] node[coordinate,pos=.6] (mou) {} (m-1-3) %
      (m-1-2) edge[celllr={0}{.1},labelar={\lambda }] (mou) %
      (mou) edge[celllr={.3}{.3},labelr={\rho }] (bbb) %
    }
  \end{center}

  \begin{proof}[Proof of Theorem~\ref{thm:facto}]
    By coherence for weak double
    categories~\cite[Theorem~7.5]{GrandisPare}, and assuming a higher
    universe in which $\CCC$ is small, we may assume that $\CCC$ is in fact
    a (strict) double category.
    
    Composition of cells in $\CCC/_\RRR M$ is just as in $\CCC$, so the only
    non-trivial point to check is weak associativity and unitality of
    horizontal composition (of morphisms).  For weak associativity, we
    observe that both cells
    \begin{center}
      \diag(.6,1){%
      A \& B \& C \& D \\
      M_v  \& M_v  \& M_v  \& M_v 
    }{%
      (m-1-1) edge[pro,twoa={S}] (m-1-2) %
      edge[labell={p}] (m-2-1) %
      (m-2-1) edge[pro,twob={M_V}] (m-2-2) %
      edge[bend right,pro,labelb={M_V}] node[coordinate] (bbb) {}
      node[coordinate,pos=.9] (trouc) {} (m-2-3) %
      edge[bend right=40,pro,labelb={M_V}] node[coordinate,pos=.6] (bidoule) {} (m-2-4) %
      (m-1-2) edge[labelr={q}] (m-2-2) %
      (m-1-2) edge[pro,twoabove={aa}{T}] (m-1-3) %
      (m-2-2) edge[pro,twobelow={bb}{M_V}] (m-2-3) %
      (m-1-3) edge[labelr={r}] (m-2-3) %
      (a) edge[cell=0,labell={\alpha }] (b) %
      (aa) edge[cell=0,labell={\beta }] (bb) %
      (m-1-3) edge[pro,twoabove={aaa}{U}] (m-1-4) %
      (m-2-3) edge[pro,twobelow={bbbb}{M_V}] (m-2-4) %
      (m-1-4) edge[labelr={s}] (m-2-4) %
      (m-2-2) edge[labelr={\ \mu },celllr={0}{0.1}] (bbb) %
      (trouc) edge[labelr={\ \mu },celllr={0.2}{0.2}] (bidoule) %
      (aaa) edge[cell=0,labell={\gamma }] (bbbb) %
    }
    \hfil
    \diag(.6,1){%
      A \& B \& C \& D \\
      M_v  \& M_v  \& M_v  \& M_v 
    }{%
      (m-1-1) edge[pro,twoa={S}] (m-1-2) %
      edge[labell={p}] (m-2-1) %
      (m-2-1) edge[pro,twob={M_V}] (m-2-2) %
      edge[bend right=40,pro,labelb={M_V}] node[coordinate,pos=.4] (bidoule) {} (m-2-4) %
      (m-2-2) edge[bend right,pro,labelb={M_V}] node[coordinate] (bbb) {}
      node[coordinate,pos=.1] (trouc) {} (m-2-4) %
      (m-1-2) edge[labelr={q}] (m-2-2) %
      (m-1-2) edge[pro,twoabove={aa}{T}] (m-1-3) %
      (m-2-2) edge[pro,twobelow={bb}{M_V}] (m-2-3) %
      (m-1-3) edge[labelr={r}] (m-2-3) %
      (a) edge[cell=0,labell={\alpha }] (b) %
      (aa) edge[cell=0,labell={\beta }] (bb) %
      (m-1-3) edge[pro,twoabove={aaa}{U}] (m-1-4) %
      (m-2-3) edge[pro,twobelow={bbbb}{M_V}] (m-2-4) %
      (m-1-4) edge[labelr={s}] (m-2-4) %
      (m-2-3) edge[labelr={\ \mu },celllr={0}{0.1}] (bbb) %
      (trouc) edge[labell={\mu \ },celllr={0.2}{0.2}] (bidoule) %
      (aaa) edge[cell=0,labell={\gamma }] (bbbb) %
    }
    \end{center}
    are equal.  Now, denoting composition in $(\CCC/_\RRR M)_h$ by
    $\tilde{\bullet}$, $\gamma  \tilde{\bullet} (\beta  \tilde{\bullet} \alpha )$ and
    $(\gamma  \tilde{\bullet} \beta ) \tilde{\bullet} \alpha $ are obtained by factoring them as
    follows.  For the former, we factor
    \begin{center}
    $T \bullet S \xto{\beta  \bullet \alpha } M_V \bullet M_V \xto{\mu } M_V$
    \hfil as \hfil
    $T \bullet S \xto{\lambda _{\beta ,\alpha }} K_{\beta ,\alpha } \xto{\rho _{\beta ,\alpha }} M_V,$
  \end{center}
  in which, by Condition~\ref{double:facto:split}, $\lambda _{\beta ,\alpha }$ has
  identity left and right borders, i.e., is \emph{special}.  We then
  factor
    \begin{center}
    $U \bullet K_{\beta ,\alpha } \xto{\gamma  \bullet \rho _{\beta ,\alpha }} M_V \bullet M_V \xto{\mu } M_V$
    \hfil as \hfil
    $U \bullet K_{\beta ,\alpha } \xto{\lambda _{\gamma ,(\beta ,\alpha )}} K_{\gamma ,(\beta ,\alpha )} \xto{\rho _{\gamma ,(\beta ,\alpha )}} M_V.$
  \end{center}
  The other composite may be computed symmetrically, so that
  we obtain factorisations:
  \begin{center}
    \diag{%
      A \& B \& C \& D \\
      M_v  \&  \&  \& M_v 
    }{%
      (m-1-1) edge[pro,twoa={S}] (m-1-2) %
      edge[labell={p}] (m-2-1) %
      edge[bend right=40,pro,labelb={}] node[coordinate,pos=.8] (trouc) {} (m-1-3) %
      edge[bend right=40,pro,labelbat={K_{\gamma ,(\beta ,\alpha )}}{0.4}]
      node[coordinate,pos=.6] (bidoule) {} (m-1-4) %
      (m-2-1)
      edge[bend right=40,pro,labelb={M_V}] node[coordinate,pos=.6] (machin) {} (m-2-4) %
      (m-1-2) edge[pro,twoabove={aa}{T}] (m-1-3) %
      (m-1-3) edge[pro,twoabove={aaa}{U}] (m-1-4) %
      (m-1-4) edge[labelr={s}] (m-2-4) %
      (m-1-2) edge[labell={\lambda _{\beta ,\alpha }\ },celllr={0}{0.1}] (trouc) %
      (trouc) edge[labelar={\lambda _{\gamma ,(\beta ,\alpha )}},celllr={0.2}{0.2}] (bidoule) %
      (bidoule) edge[labelr={\ \rho _{\gamma ,(\beta ,\alpha )}},celllr={0.2}{0.2}] (machin) %
    }
    \hfil
    \diag{%
      A \& B \& C \& D \\
      M_v  \&  \&  \& M_v \rlap{.}
    }{%
      (m-1-1) edge[pro,twoa={S}] (m-1-2) %
      edge[labell={p}] (m-2-1) %
      (m-1-2) edge[bend right=40,pro,labelb={}] node[coordinate,pos=.2] (trouc) {} (m-1-4) %
      (m-1-1)
      edge[bend right=40,pro,labelbrat={K_{(\gamma ,\beta ),\alpha }}{0.7}]
      node[coordinate,pos=.4] (bidoule) {} (m-1-4) %
      (m-2-1)
      edge[bend right=40,pro,labelb={M_V}] node[coordinate,pos=.4] (machin) {} (m-2-4) %
      (m-1-2) edge[pro,twoabove={aa}{T}] (m-1-3) %
      (m-1-3) edge[pro,twoabove={aaa}{U}] (m-1-4) %
      (m-1-4) edge[labelr={s}] (m-2-4) %
      (m-1-3) edge[labelr={\ \ \ \lambda _{\beta ,\gamma }},celllr={0}{0.1}] (trouc) %
      (trouc) edge[labelal={\lambda _{(\gamma ,\beta ),\alpha }},celllr={0.2}{0.2}] (bidoule) %
      (bidoule) edge[labelr={\ \rho _{(\gamma ,\beta ),\alpha }},celllr={0.2}{0.2}] (machin) %
    }
  \end{center}
  By Condition~\ref{double:facto:L}, both are in fact factorisations
  for $(\LLL_V,\RRR _V)$, so that by lifting, we obtain a special cell
  $a_{\alpha ,\beta ,\gamma }\colon  K_{\gamma ,(\beta ,\alpha )} \isoto K_{(\gamma ,\beta ),\alpha }$ such that
  $\rho _{(\gamma ,\beta ),\alpha } \circ  a_{\alpha ,\beta ,\gamma } = \rho _{\gamma ,(\beta ,\alpha )}$, which is our candidate
  associator for $\CCC/_\RRR M$. It satisfies the MacLane pentagon by
  uniqueness of lifting.

  Weak unitality follows similarly.
\end{proof}

We finally obtain:
\begin{cor}
  Consider the weak double category $\Span(\Cat)/_{\dfib}\H$.
  Restricting horizontal morphisms to discrete fibrations
  $S \to  \P _{A,B}$ that are subcategory inclusions, we obtain a category
  which is isomorphic to the standard category of simple games.
\end{cor}

\subsection{Day convolution}
We finally reach the surprising application mentioned in the
introduction, Day convolution.  The purpose of this operation is to
show that the Yoneda embedding $\yoneda \colon  \C  \to  \psh{\C }$ is monoidal when $\C $
is.  This means that $\psh{\C }$ may be equipped with monoidal structure
preserved by $\yoneda $ up to coherent isomorphism.  The tensor is given as
follows:
\begin{defi}
  For any small monoidal category $\C $ and $X,Y \in  \psh{\C }$, let
  $$(X \otimes  Y)(c) = \int ^{(c_1 ,c_2 ) \in  \C ^2 } X(c_1 ) \times  Y(c_2 ) \times  \C (c,c_1 \otimes c_2 ).$$
\end{defi}

Let us now recover this structure from Theorem~\ref{thm:facto}, in the
particular case where $\C $ is strictly monoidal.  The starting point is
the sub weak double category, say $\W $, of $\Span(\Cat)$, obtained by
restricting attention to just one object and one vertical morphism,
namely the terminal category $\one$ and the identity thereon. Thus, $\W _V$
consists of categories and functors, and horizontal composition is
given by cartesian product. Furthermore, a monad in $\W $ is nothing but
a monoid in $\Cat$ for the cartesian product, i.e., a strict monoidal
category $\C $.  The double factorisation system induced by discrete
fibrations on $\Span(\Cat)$ restricts to one on $\W $, and obviously the
weak double category $\W /_{\dfib}\C $ is vertically trivial, hence
underlies a monoidal category, say $\C '$.

\begin{thm}\label{thm:convolution}
  For any strictly monoidal category $\C $, the monoidal category $\C '$
  is equivalent to $\psh{\C }$ equipped with the convolution tensor
  product.
\end{thm}

In order to prove this, let us first show:
\begin{lem}\label{lem:comprehensive}
  Let $f\colon  A \to  B$ be a functor. The discrete fibration $\rho _f$ associated
  to $f$ is determined up to isomorphism by
  $$\partial ^\star (\rho _f)(b) \iso  \int ^{a \in  A} B(b,f(a)),$$
  where $\partial ^\star \colon  \dfib_B \to  \psh{B}$ is the standard equivalence between
  discrete fibrations and presheaves.
\end{lem}

\begin{proof}
  This is actually obvious by construction. In~\cite{comprehensive},
  the dual case is actually treated, initial functors and discrete
  opfibrations.  But up to this discrepancy, $\partial ^\star (f)$ is precisely $k$
  in the proof of~\cite[Theorem 3]{comprehensive}, which would in our
  case be defined as the left Kan extension of
  $\op{A} \xto{!} 1 \xto{\name{1}} \Set $ along $\op{f}$.
  By the well-known characterisation of left Kan extensions by coends,
  we readily obtain the desired formula.
\end{proof}

\begin{proof}[Proof of Theorem~\ref{thm:convolution}]
  By construction, given two presheaves $X,Y \in  \psh{\C }$ and
  transporting them to their corresponding discrete fibrations, say
  $S\colon  el(X) \to  \C $ and $T\colon  el(Y) \to  \C $, their tensor product $S \tilde{\bullet} T$ in $\C '$ is
  the right factor of the composite
  $$el(X) \times  el(Y) \xto{S \times  T} \C  \times  \C  \xto{\otimes } \C .$$
  By Lemma~\ref{lem:comprehensive}, the result has its corresponding
  presheaf defined up to isomorphism by
  \begin{center}
  \hfill $\begin{array}[b]{rcl}
      \partial ^\star (S \tilde{\bullet} T)(c) & \iso  & \int ^{(a,b) \in  el(X) \times  el(Y)} \C (c, \otimes  ((S\times T)(a,b))) \\
                            & = & \int ^{(a,b) \in  el(X) \times  el(Y)} \C (c, S(a) \otimes  T(b)) \\
                  & \iso  & \int ^{((c_1 ,x),(c_2 ,y)) \in  el(X) \times  el(Y)} \C (c, c_1 \otimes c_2 ) \\
                  & \iso  & \int ^{c_1 ,c_2 } X(c_1 ) \times  Y(c_2 ) \times  \C (c, c_1 \otimes c_2 ), \mbox{as desired.} 
    \end{array}$ \hfill \qedhere
  \end{center}
  \end{proof}

  \section{Conclusion and perspectives}\label{sec:conc}
  We have designed an abstract slice construction over monads in weak
  double categories, which has as instances
  \begin{itemize}
  \item a weak double category of simple games and concurrent
    strategies,
  \item and the monoidal category of presheaves over any strict
    monoidal category.
  \end{itemize}
  We see at least two directions for future work. First, we should try
  to accomodate not only the weak double category structure of
  Melliès's construction, but also symmetric monoidal
  closedness. Melliès is also currently working on the construction of
  a linear exponential comonad~\cite{DBLPconf/csl/BentonBPH92} on his
  category of simple games and concurrent strategies. This will of
  course be a useful feature to incorporate to our framework.  The
  second direction for future work is to generalise our contsruction
  to encompass Day convolution for non-strict monoidal
  categories. This will involve a 3-dimensional refinement of weak
  double categories.

\bibliographystyle{plainnat}
  
  \bibliography{../bib}
\end{document}